\documentclass[12pt]{amsart}

\usepackage{anysize}
\usepackage{hyperref}
\marginsize{2cm}{2cm}{2cm}{2cm}

\usepackage{natbib}
\usepackage[foot]{amsaddr}

\usepackage{graphicx}
\usepackage{caption}
\usepackage{subcaption}
\usepackage{setspace}

\newtheorem{thm}{Theorem}

\newtheorem{lem}[thm]{Lemma}
\newtheorem{prop}[thm]{Proposition}
\theoremstyle{definition}
\newtheorem{defn}[thm]{Definition}
\theoremstyle{example}

\theoremstyle{remark}
\newtheorem{rem}[thm]{Remark}

\numberwithin{thm}{section}
\numberwithin{equation}{section}

\newcommand{\Acc}{\mathcal{A}}
\newcommand{\R}{\mathbb{R}}
\newcommand{\E}{\mathbb{E}}
\newcommand{\Prob}{\mathbb{P}}
\newcommand{\F}{\mathcal{F}}
\newcommand{\cadlag}{c\`{a}dl\`{a}g }
\newcommand{\Rspace}{\mathcal{R}}

\newcommand{\norm}[2]{\| #1 \|_{#2}}

\newcommand{\A}{\mathcal{A}}

\newcommand{\DT}{\mathrm{DT}}
\newcommand{\CED}{\mathrm{CED}}

\newcommand{\TM}{\mathrm{TM}}

\begin{document}

\title{The Temporal Dimension of Risk \\ \vspace{5pt} (\lowercase{\normalfont forthcoming in} {\emph{\scshape \itshape \small Journal of Risk}})} 
\author{Ola Mahmoud$^1$}
\address{$^1$Faculty of Mathematics and Statistics, University of St. Gallen (Switzerland) and Center for Risk Management Research, University of California, Berkeley (USA)}
\email{ola.mahmoud@unisg.ch, olamahmoud@berkeley.edu}
\date{ \today}
\maketitle

\begin{abstract}
Multi-period measures of risk account for the path that the value of an investment portfolio takes. In the context of probabilistic risk measures, the focus has traditionally been on the magnitude of investment loss and not on the dimension associated with the passage of time. In this paper, the concept of temporal path-dependent risk measure is mathematically formalized to capture the risk associated with the temporal dimension of a stochastic process and its theoretical properties are analyzed. We then study the temporal dimension of investment drawdown, its duration, which measures the length of excursions below a running maximum. Its properties in the context of risk measures are analyzed both theoretically and empirically. In particular, we show that duration captures serial correlation in the returns of two major asset classes. We conclude by discussing the challenges of path-dependent temporal risk estimation in practice.
\end{abstract}

\section{Introduction}

Single-period measures of risk do not account for the path an investment portfolio takes. Since investment funds do not hold static positions, measuring the risk of investments should ideally be defined over random paths rather than random single-period gains or losses. Mathematically, a path-dependent measure of risk is a real valued function $\rho:\Rspace^\infty\to\R$ on the space of stochastic processes $\Rspace^\infty$ representing cumulative returns over a path of fixed length. Most existing path-dependent risk measures are essentially a measure of the spatial dimension of risk, that is the magnitude of investment loss or gain. However, by moving from the single-period to the multi-period framework, a second dimension becomes manifest, namely that of time. This temporal dimension to a stochastic process has traditionally not been incorporated into the probabilistic theory of risk measures pioneered by \cite{Artzner1999}. 

In this paper, we formalize the temporal dimension of multi-period risk as a \emph{temporal risk measure}.
For a given time horizon $T\in(0,\infty)$, a temporal risk measure is a path-dependent risk measure $\rho:\Rspace^\infty\to\R$, which first maps a stochastic process $X\in\Rspace^\infty$ to a random time $\tau$, which is a random variable taking values in the time interval $[0,T]$. This so-called temporal transformation is shift and scaling invariant and is hence invariant to the spatial dimension of the stochastic process. The random variable $\tau$ is meant to summarize a certain temporal behavior of the process $X$ that we are interested in. Then, a real-valued risk functional, such as deviation or tail mean, is applied to $\tau$, describing a risky feature of its distribution. We derive some properties of temporal risk measures and show that they are not coherent in the sense of \cite{Artzner1999}.

The second part of the paper focuses on one of the most widely quoted indicators of multi-period risk: drawdown, which is the decline from a historical peak in net asset value or cumulative return. In the event of a large drawdown, conventional single-period risk diagnostics, such as volatility or Expected Shortfall, are irrelevant and liquidation under unfavorable market conditions after an abrupt market decline may be forced. Since the notion of drawdown inherently accounts for the path over a given time period, it comes equipped with two dimensions: a spatial dimension (drawdown magnitude) and a temporal dimension (drawdown duration). 
While the magnitude component of drawdown has been extensively studied in the academic literature and is regularly used by the investment community, the temporal dimension, its duration, which measures the length of excursions below a running maximum, has not received the same kind of attention. In particular, even though it is a widely quoted  performance measure, a generally accepted mathematical methodology for forming expectations about future duration does not seem to exist in practice. 

To this end, we analyze the properties of drawdown duration theoretically, in the context of temporal risk measures, and empirically, by looking at some empirical duration distributions. We also show that duration risk is highly sensitive to serial correlation in asset returns and hence captures their temporal dependence. This insight may impact how certain portfolio construction strategies are perceived. For example, the returns of the popular momentum strategy (\cite{Chan1996}) are highly autocorrelated. Our work implies that such strategies would suffer from high duration risk.

In summary, the main contribution of this work is twofold:
\begin{itemize}
\item[(i)] First, we formalize the theory of temporal risk measures and analyze their properties.  We thereby introduce a new risk diagnostic complementing traditional ones, uniquely capturing the risk associated with the passage of time, and providing more path-wise information than standard risk measures. By incorporating the time dimension into the framework of risk measurement, one can in practice form expectations about future temporal risk. 
\item[(ii)] Second, we illustrate a practical example applying our theory of temporal risk measures. More specifically, we study the temporal dimension of drawdown, its duration. Drawdown duration is a widely quoted risk diagnostic in the investment management industry but has not been studied before in the context of path-dependent measures of risk. We hence formulate duration as a temporal measure of risk and verify its properties. We then derive some empirical properties of duration risk. In particular, we show that duration captures serial correlation in the returns of two major asset classes. 
\end{itemize}

\subsection{Synopsis}
We start in Section \ref{section:temporal_risk} by reviewing the probabilistic theory of path-dependent risk measures in a continuous-time setting. We then introduce the notion of temporal transformation, a spatial invariant random variable mapping of a stochastic process to a random time, and the notion of temporal acceptance family and show that these two constructs correspond bijectively. A temporal risk measure is then defined as a path-dependent risk measure which can be decomposed into a temporal transformation and a risk functional. We show that temporal measures of risk can never be coherent in the sense of \cite{Artzner1999}.
In Section \ref{section:duration}, the temporal dimension of drawdown is analyzed. We first review the spatial dimension of drawdown, its magnitude. Its temporal dimension, duration, captures the time it takes a stochastic process to reach a previous running maximum for the first time. Section \ref{section:duration_risk} includes an analysis of duration in terms of temporal measures of risk and an empirical analysis of the distribution of duration. We then show that duration risk captures temporal dependence in terms of serial correlation to a greater degree than traditional one-period risk measures. We conclude in Section \ref{section:conclusion} with a summary of our findings and discuss the challenges of path-dependent risk estimation in practice.

\subsection{Background literature}
We summarize work related to the probabilistic theory of path-dependent risk measures, and to the theoretical and practical analysis of the two dimensions of drawdown, its magnitude and duration.

\subsubsection{Path-dependent risk measures}
The seminal work of \cite{Artzner1999} introducing coherent risk measures is centered around single-period risk, where risk is measured at the beginning of the period and random loss or gain is observed at the end of the period. In \cite{Artzner2002, Artzner2007}, the framework of \cite{Artzner1999} is generalized to discrete-time multi-period models, and in \cite{Cheridito2004, Cheridito2005} representation results for coherent and convex risk measures were developed for continuous-time stochastic models. \cite{Riedel2004} defines the concept of dynamic risk measure, where dynamic risk assessment consists of a sequence of risk mappings and is updated as time evolves to incorporate new information. Such measures come equipped with a notion of dynamic consistency, which requires that judgements based on the risk measure are not contradictory over time (see also \cite{Bion2008, Bion2009} and \cite{Fasen2012}). Dynamic risk measures have been studied extensively over the past decade; see \cite{FoellmerPenner2006}, \cite{Cheridito2006}, \cite{Kloeppel2007}, and \cite{Fritelli2004}, amongst others.

We point out that the focus of the studies mentioned above is on the magnitude of losses and gains and not on the temporal behavior of the underlying process. To our knowledge, the notion of path-dependent risk measure capturing the temporal dimension of risk has not been formally developed in the academic literature.

\subsubsection{Drawdown magnitude and duration}
The analytical assessment of drawdown magnitudes has been broadly studied in the literature of applied probability theory (\citet{Taylor1975}, \citet{Lehoczky1977}, \citet{Douady2000}, \citet{Magdon2004}, \citet{Landriault2015a}, \citet{Mijatovic2012}, \citet{Zhang2010}, \citet{Hadjiliadis2006}, \citet{Pospisil2009}). The reduction of drawdown in active portfolio management has received considerable attention in mathematical finance research (\citet{Grossman1993}, \citet{Cvitanic1995}, \citet{Chekhlov2003, Chekhlov2005}, \citet{Krokhmal2003}, \citet{Carr2011},  \citet{Cherney2013}, \citet{Sekine2013}, \citet{Zhang2013}, \citet{Zhang2015},
\citet{Uryasev2014a}, \citet{Pospisil2010}). 
In the context of probabilistic risk measurement, which is our main interest in this paper, \citet{Chekhlov2003, Chekhlov2005} develop a quantitative measure of drawdown risk called Conditional Drawdown at Risk (CDaR), and \cite{GoldbergMahmoud2014a} develop a measure of maximum drawdown risk called Conditional Expected Drawdown (CED). Both risk measures, CDaR and CED, are deviation measures (\citet{Rockafellar2002, Rockafellar2006}).
The temporal dimension of drawdown, its  duration, has not been previously studied in the context of risk measures. However, it has been considered in terms of its probabilistic properties. For example,  
in \cite{Zhang2012}, the probabilistic behavior of drawdown duration is analyzed and the joint Laplace transform of the last visit time of the maximum of a process preceding the drawdown, its, and the maximum of the process under general diffusion dynamics is derived. More recently, \cite{Landriault2015b} consider derive the duration of drawdowns for a large class of Levy processes and find that the law of duration of drawdowns qualitatively depends on the path type of the spectrally negative component of the underlying Levy process.


\section{Temporal measures of risk}
\label{section:temporal_risk}

We review the notion of path-dependent measure of risk in the continuous-time setting and formalize the mathematical framework surrounding path-dependent temporal risk measures.

\subsection{Path-dependent risk measures}
In classical risk assessment, uncertain portfolio outcomes over a fixed time horizon are represented as random variables on a probability space. A risk measure maps each random variable to a real number summarizing the overall position in risky assets of a portfolio.
Such single-period measures of risk do not account for the return path an investment portfolio takes. Since investment funds do not hold static positions, measuring the risk of investments should ideally be defined over random processes rather than random variables.

We use the general setup of \citet{Cheridito2004} for the mathematical formalism of continuous-time path dependent risk. Continuous-time cumulative returns, or equivalently net asset value processes, are represented by essentially bounded \cadlag processes (in the given probability measure) that are adapted to the filtration of a filtered probability space. More formally, for a time horizon $T \in (0,\infty)$, let $(\Omega, \F, \{\F_t\}_{t\in[0,T]},\Prob)$ be a filtered probability space satisfying the usual assumptions, that is the probability space $(\Omega, \F,\Prob)$ is complete, $(\F_t)$ is right-continuous, and $\F_0$ contains all null-sets of $\F$. For $p \in [1,\infty]$,  $(\F_t)$-adapted \cadlag processes lie in the Banach space
$$ \Rspace^p = \left\{  X \colon [0,T] \times \Omega \to \R \mid X  \hspace{5pt} (\F_t)\textrm{-adapted \cadlag process }, \norm{X}{\Rspace^p} \right\}, $$
which comes equipped with the norm
$$ \norm{X}{\Rspace^p} := \norm{X^*}{p} $$
where $X^* = \sup_{t\in[0,T]} |X_t|$.

All equalities and inequalities between processes are understood throughout in the almost sure sense with respect to the probability measure $\Prob$. For example, for processes $X$ and $Y$, $X \leq Y$ means that for $\Prob$-almost all $\omega \in \Omega$, $X_t(\omega) \leq Y_t(\omega)$ for all $t$.

\begin{defn}[Continuous-time path-dependent risk measure]
A continuous-time \emph{path-dependent risk measure} is a real-valued function $\rho: \Rspace^\infty \to \R$.
\end{defn}

Analogous to single period risk, a path-dependent risk measure $\rho \colon \Rspace^\infty \to \R$ is \emph{monetary} if it satisfies the following axioms:
\begin{itemize} 
\item Translation invariance: For all $X \in \Rspace^\infty$ and all constant almost surely $C \in \Rspace^\infty$, \linebreak $\rho(X+C) = \rho(X)-C$.
\item Monotonicity: For all $X, Y \in \Rspace^\infty$ such that $X \leq Y$, $\rho(X) \leq \rho(Y)$.\end{itemize}
It is positive homogenous of degree one if for all $X \in \Rspace^\infty$ and $\lambda > 0$, $\rho(\lambda X) = \lambda \rho(X)$; and it is convex if for all $X, Y \in \Rspace^\infty$ and $\lambda \in [0,1]$, $\rho(\lambda X + (1-\lambda)Y) \leq \lambda \rho(X) + (1-\lambda) \rho(Y)$.  A monetary (path-dependent) risk measure that is both positive homogenous and convex is called \emph{coherent}.

\begin{rem}[Notational conventions]
For a stochastic process $X\in\Rspace^\infty$, we will write $\overline{X}$ for the stochastic process defined by $\overline{X}_t:=\sup_{u\in[0,t]} X_u$, which is the running maximum of $X$ up to time $t$. Moreover, for a random variable $Z$ and confidence level $\alpha\in[0,1]$, we will use 
$$ Q_\alpha(Z) := \inf \{ d\in\R : \Prob(Z>d) \leq 1-\alpha \} $$
to denote its $\alpha$-quantile, and
$$ TM_\alpha(Z) := \frac{1}{1-\alpha} \int_\alpha^1 Q_u(Z)\mathrm{d}u$$
to denote its $\alpha$-tail-mean.
\end{rem}

\subsection{Towards temporal path-dependent risk}

Most path-dependent risk measures are in essence a measure of the spatial dimension of the investment path, that is the magnitude of investment loss or gain. One may, however, be interested in the risk associated with the temporal dimension of the underlying stochastic process. For example in the fund management industry, historical values of the time it takes to regain a previous maximum (``peak-to-peak"), or the length of time between a previous maximum and a current low (``peak-to-trough") are frequently quoted alongside drawdown values. However, a generally accepted mathematical methodology for forming expectations about future potential such temporal risks does not seem to exist. 

Whereas the spatial dimension of a process summarizes a monetary quantity, such as investment gains, losses, or returns, the temporal dimension is measured in time units. It is therefore expected that the two quantities do not behave in the same way. To transition from the spatial dimension of a process to its temporal dimension, we define the so-called \emph{temporal transformation} which essentially rids a process of its spatial dimension by mapping it to a \emph{random time}. We formalize these notions next.

Fix a time horizon $T\in(0,\infty)$. Given a stochastic process $X\in\Rspace^\infty$, a \emph{random time} $\tau$ is a random variable on the same probability space $(\Omega, \F,\Prob)$ as $X$, taking values in the time interval $[0,T]$. We say that $X_\tau$ denotes the state of the process $X$ at random time $\tau$. Random times can be thought of as elements of the space $\mathcal{T} \subset L_0(\Omega, \F,\Prob)$ of real-valued random variables $\tau:\Omega\to[0,T]$. Note that the space $\mathcal{T}$ is a partially ordered set; it is reflexive, antisymmetric and transitive. The vector order $\leq$ on $\mathcal{T}$ is given by $\tau_1\leq\tau_2$ if and only if $\tau_1(\omega) \leq \tau_2(\omega)$ for almost all $\omega\in\Omega$.
In the probabilistic study of stochastic processes, typical examples of random time include the \emph{hitting time}, which is the first time at which a given process hits a given subset of the state space, and the \emph{stopping time}, which is  the time at which a given stochastic process exhibits a certain behavior of interest.

\subsection{Temporal transformations}
Given a stochastic process $X\in\Rspace^\infty$, we are now interested in the time it takes for certain events to occur. To extract this temporal trait of the process, we use a spatial invariant transformation to map $X$ to a random time in $\mathcal{T}$.

\begin{defn}[Temporal transformation]
A random variable transformation $\theta:\Rspace^\infty\to\mathcal{T}$ is called a \emph{temporal transformation} if it satisfies the following three axioms:
\begin{enumerate}
\item Normalization: For all constant deterministic $C\in\Rspace^\infty$, $\theta(C) = 0$.
\item Shift invariance: For all $X\in\Rspace^\infty$ and constant deterministic $C$, $\theta(X+C) = \theta(X)$.
\item Scaling invariance: For all $X\in\Rspace^\infty$ and $\lambda>0$, $\theta(\lambda X) = \theta(X)$.
\end{enumerate}
\end{defn}

The invariance under changes in the magnitude of a process, implied by axioms (2) and (3) above, essentially yields a random variable which is in some way independent of the spatial traits of the original process. A temporal transformation constructs such a random time summarizing a temporal property that we are interested in. 
The invariance axioms are substantially all that is required to obtain a temporal trait, as they discard any spatial features. Moreover, note that a temporal transformation need not satisfy the monotonicity axiom --- a given stochastic process being larger than another gives no indication of their ordering with respect to their respective temporal characteristics. 

The random time $\theta(X)$ associated with a stochastic process $X\in\Rspace^\infty$ can be interpreted as the time it takes $X$ to experience a certain risky property. Under this interpretation, one may think of longer periods of time to be worse than shorter ones, for example the time it takes to recover a previous high. With this in mind, a temporal transformation $\theta:\Rspace^\infty\to\mathcal{T}$ induces the following \emph{temporal preference relation}\footnote{A preference relation on a set $A$ is a binary relation $\succeq$ satisfying asymmetry and negative transitivity. It induces an indifference relation $\sim$ defined by $a\sim b$ if and only if $a\succeq b$ and $b\succeq a$ for $a,b\in A$.} $\succeq_\theta$ on $\Rspace^\infty$:

\begin{center}
For $X,Y\in\Rspace^\infty$,  $\quad X\succeq_\theta Y \quad $ if and only if $\quad \theta(X)\leq\theta(Y)$.
\end{center} 

This temporal preference order captures the idea that a process whose risky feature lasts longer in every state of the world is less preferred. Note in particular that the relation $\succeq_\theta$ satisfies $X+C\sim X$ and $\lambda X\sim X$.

To every temporal transformation $\theta$ and associated preference order $\succeq_\theta$ we associate its \emph{temporal acceptance family}  $\A_\theta$, a concept introduced by \cite{DrapeauKupper2013} in the context of single-period risk measures, generalizing the notion of risk acceptance set of \cite{Artzner1999}. 

\begin{defn}[Temporal acceptance family]
An increasing family $\Acc = (\Acc^\tau)_{\tau\in\mathcal{T}}$ of subsets $\Acc^\tau \subseteq \Rspace^\infty$ is a \emph{temporal acceptance family} if it satisfies the following properties:
\begin{enumerate}
\item For all constant deterministic $C\in\Rspace^\infty$ and $\tau\in\mathcal{T}$, $\Acc^\tau +C = \Acc^\tau$.
\item For all $\tau\in\mathcal{T}$ and $\lambda>0$, $\lambda\Acc^\tau = \Acc^\tau$.
\end{enumerate}
\end{defn}

Acceptance sets were traditionally introduced as an instrument for robust representation results of risk measures, and are often used to derive structural properties of risk measures and to model certain economic principles of risk. In our setup, a given temporal acceptance set abstractly represents all processes that share a specific temporal characteristic.


Indeed, there is a bijective correspondence between temporal transformations and temporal acceptance sets, the verification of which is straightforward.

\begin{prop}
Let $\theta:\Rspace^\infty\to\mathcal{T}$ be a temporal transformation with associated temporal preference order $\succeq_\theta$. Then the family $\Acc_\theta$ of subsets $\Acc_\theta^\tau \subseteq\Rspace^\infty$ given by
$$ \Acc^\tau_\theta = \{ X\in\Rspace^\infty \ \colon \ \theta(X) \leq \tau \} $$
is a temporal acceptance family. Conversely, given a temporal acceptance family $\Acc^\tau$, the transformation $ \theta_\Acc : \Rspace^\infty \to \mathcal{T} $ defined by
$$ \theta_\Acc (X) := \inf \{ \tau\in\mathcal{T} \ : \ X\in\Acc^\tau \} $$
is a temporal transformation. Moreover, this correspondence is bijective, that is $\theta = \theta_{\Acc_\theta}$ and $\Acc = \Acc_{\theta_\Acc}$.
\end{prop}

\subsection{Temporal measures of risk}

Now that we have transformed a stochastic process into a random variable representing some temporal characteristic, we can examine the distribution of this random time using a risk functional.
A temporal risk measure\footnote{The notion of temporal risk we introduce is not to be confused with that of \cite{Machina1984}, which, in the context of economic utility maximizing preferences,  captures the idea of delayed risk as opposed to immediately resolved risk when choosing amongst risky prospects.} is a path-dependent measure of risk which summarizes a temporal property of a stochastic process.
It is essentially a real-valued function describing a feature of the distribution of a random time associated with the stochastic process. Formally, temporal measures of risk can be decomposed into a temporal transformation and a risk measure:

\begin{defn}[Temporal risk measure]
For a given time horizon $T\in(0,\infty)$, a \emph{temporal risk measure} is a path-dependent risk measure $\rho_T: \Rspace^\infty \to \R$ defined by 
$$ \rho_T := \rho \circ \theta \ ,$$
where $\theta : \Rspace^\infty \to \mathcal{T}$ is a temporal transformation mapping a stochastic process to a random time, and $\rho:\mathcal{T}\to\R$ is a real-valued risk functional.
\label{defn:temporal}
\end{defn}

Note that there are no restrictions on either the risk functional $\rho$ or the temporal risk measure $\rho_T$. In particular, they need not satisfy the conventional properties of measures of risk, as we are not interested in the monetary dimension of a stochastic process. However, we can investigate the conditions under which we do obtain a coherent measure of temporal risk $\rho_T$. More specifically, since $\rho_T$ is a composite of a risk functional with a temporal transformation, the effect of composing a coherent risk functional $\rho$ with temporal transformations does not yield coherent temporal measures.

\begin{lem}
Path-dependent temporal measures of risk $\rho_T:\Rspace^\infty\to\R$ are not coherent measures of risk. In particular, given that $\rho_T=\rho\circ\theta$, coherence of the risk functional $\rho:\mathcal{T}\to\R$ does not imply coherence of $\rho_T$.
\end{lem}

\begin{rem}[Impact of non-coherence in practice]
The axioms of coherence were originally introduced as desirable properties for risk measures under the assumption that the risk of a position represents the amount of capital that should be added so that it becomes acceptable to the regulator. For example, from a regulatory viewpoint, translation invariance means that adding the value of any guaranteed position to an existing portfolio simply decreases the capital required by that guaranteed amount; monotonicity essentially states that positions that lead to higher losses should require more risk capital. From the temporal perspective, the monetary unit, and hence these two axioms, are irrelevant. On the other hand, convexity and positive homogeneity are two practically useful properties under any dimension --- convexity enables investors to allocate funds in such a way that minimizes overall risk, while positive homogeneity ensures that the overall risk of a portfolio can be linearly decomposed into additive subcomponents representing the individual factor contributions to risk. Temporal measures of risk are neither convex nor degree-one positive homogenous. This implies on the one hand that linear attribution to random time risk is not supported, and on the other hand that the favorable convex optimization theory is not applicable. Temporal risk hence seems to have limited practical application in the investment process.
Despite this, temporal risk encapsulates a potentially useful diagnostic measure of a dimension of risk that is traditionally not incorporated in the risk management process within the investment management industry. In addition to being a new risk diagnostic and assuming a realistic and efficient risk model, temporal risk measures enable the investor to form expectation about future potential temporal risk in practice.
\end{rem}


\section{The Temporal Dimension of Drawdown}
\label{section:duration}

We analyze the properties of the temporal dimension of drawdown, one of the most widely quoted path-dependent measures of risk in practice. The two temporal dimensions we consider are \emph{drawdown duration}, which measures the time it takes a process to reach a previous running maximum, and \emph{liquidation stopping time}, which captures a subjectively set time threshold beyond which an investor liquidates if the drawdown exceeds this threshold. We start by recalling the properties of the magnitude of drawdown.

\subsection{The spatial dimension of drawdown}

Unlike conventional measures of risk, such as volatility, Value-at-Risk and Expected Shortfall, the notion of drawdown is inherently path-dependent. One of the most frequently quoted indicators of downside risk in the fund management industry, it measures the decline in value from the running maximum (high water mark) of a stochastic process representing the net asset value of an investment. 

\begin{defn}[Drawdown process]
For a horizon $T\in(0,\infty)$, the \emph{drawdown process} $D^{(X)}:=\{D^{(X)}_t\}_{t\in[0,T]}$ corresponding to a stochastic process $X\in\Rspace^\infty$ is defined by  
$$ D^{(X)}_t = \overline{X}_t - X_t \ ,$$
where 
$$ \overline{X}_t = \sup_{u\in[0,t]} X_u$$
is the running maximum of $X$ up to time $t$.
\end{defn}

The drawdown process associated with a given stochastic process has some practically intuitive properties. Clearly, a constant deterministic process does not experience any changes in value over time, implying that the corresponding drawdown process is zero. Moreover, any constant shift in a given process does not alter the magnitude of its drawdowns, and any constant multiplier of the stochastic process affects the drawdowns by the same multiplier. However, drawdown magnitudes are not preserved under monotonicity, which means that processes that can be ordered according to their magnitudes do not necessarily imply the same or opposite ordering on the drawdown magnitudes. Finally, a practically important property is convexity. Indeed, a convex combination of two processes results in a drawdown process that is smaller in magnitude than the average standalone drawdowns of the underlying processes. We formalize these properties next.

\begin{lem}[Properties of drawdown]
Given the stochastic process $X\in\Rspace^\infty$, let $D^{(X)}$ be the corresponding drawdown process for a fixed time horizon $T$. Then:
\begin{enumerate}
\item For all constant deterministic processes $C \in \Rspace^\infty$, $D^{(C)} = 0$.
\item For constant deterministic $C\in\Rspace^\infty$, $D^{(X+C)} = D^{(X)}$.
\item For $\lambda >0$, $D^{(\lambda X)} = \lambda D^{(X)}$.
\item For $Y\in\Rspace^\infty$ and $\lambda\in[0,1]$, $D^{(\lambda X+(1-\lambda)Y)}\leq \lambda D^{(X)} + (1-\lambda)D^{(Y)}$.
\end{enumerate}
\label{lem:drawdown_prop}
\end{lem}

\begin{proof}
Properties (1) through (3) are straightforward. To derive (4), note that for $\lambda\in[0,1]$, we clearly have $\overline{\lambda X+(1-\lambda)Y}\leq \lambda \overline{X} + (1-\lambda)\overline{Y}$ by properties of the supremum, and therefore $D^{(\lambda X+(1-\lambda)Y)} = \overline{\lambda X+(1-\lambda)Y} -\lambda X -(1+\lambda)Y \leq \lambda \overline{X} + (1-\lambda)\overline{Y}-\lambda X -(1+\lambda)Y = \lambda D^{(X)} + (1-\lambda)D^{(Y)}$.
\end{proof}

\begin{rem}
Note that it is not generally the case that for $Y\in\Rspace^\infty$ for which $X\leq Y$, either $D^{(X)}\leq D^{(Y)}$ or $D^{(X)}\geq D^{(Y)}$. The only thing $X\leq Y$ implies is that $\overline{X}\leq \overline{Y}$. However, since at any point in time within the horizon the magnitude of a drop from peak is not specified, one cannot form a conclusion about the magnitude order of the corresponding drawdown processes. To  see this more formally, note that under either of the assumptions that $D^{(X)}>D^{(Y)}$ or $D^{(X)}<D^{(Y)}$, we always get $X>Y$.

\end{rem}

In practice, the use of the maximum drawdown as an indicator of risk is particularly popular in the universe of hedge funds and commodity trading advisors, where maximum drawdown adjusted performance measures, such as the Calmar ratio, the Sterling ratio and the Burke ratio, are frequently used. 

\begin{defn}[Maximum drawdown]
Within a fixed time horizon $T \in (0,\infty)$, the \emph{maximum drawdown} of the stochastic process $X\in\Rspace^\infty$ is the maximum drop from peak to trough of $X$ in $[0,T]$, and hence the largest amongst all drawdowns $D^{(X)}_t$:
$$ \mu (X) : = \sup_{t\in[0,T]}\{D^{(X)}_t\} .$$
Alternatively, maximum drawdown can be defined as the random variable obtained through the following transformation of the underlying stochastic process $X$:
$$\mu(X) := \sup_{t\in[0,T]} \sup_{s\in[t,T]} \left\{ X_s - X_t\right\} \ .$$
\end{defn}

The tail of the maximum drawdown distribution, from which the likelihood of a drawdown of a given magnitude can be distilled, could be of particular interest in practice. The drawdown risk measure defined in \citet{GoldbergMahmoud2014a} is the tail mean of the maximum drawdown distribution.
At confidence level $\alpha \in [0,1]$, the \emph{Conditional Expected Drawdown} $\CED_\alpha: \Rspace^\infty \to \R$ is defined to be the path-dependent risk measure mapping the process $X$ to the expected maximum drawdown $\mu(X)$ given that the maximum drawdown threshold at $\alpha$ is breached. More formally, 
$$ \CED_\alpha (X) : = \TM_\alpha(\mu(X)) =  \frac{1}{1-\alpha} \int_\alpha^1 Q_u\left( \mu(X) \right) d u ,$$
where $Q_\alpha$ is a quantile of the maximum drawdown distribution:
$$ Q_\alpha \left( \mu(X) \right)= \inf \left\{ d\in\R :  \Prob \left( \mu(X) > d \right) \leq 1-\alpha \right\} $$
If the distribution of $\mu(X)$ is continuous, then $\CED_\alpha$ is equivalent to the tail conditional expectation:
$$ \CED_\alpha (X) = \E \left(\mu(X) \mid \mu(X) > \DT_\alpha\left( \mu(X) \right)\right).$$

CED has sound mathematical properties making it amenable to the investment process. Indeed, it is a degree-one positive homogenous risk measure, so that it can be attributed to factors, and convex, so that it can be used in quantitative optimization. 

\begin{prop}[\cite{GoldbergMahmoud2014a}]
For a given confidence level $\alpha\in[0,1]$, Conditional Expected Drawdown $\CED_\alpha:\Rspace^\infty\to\R$ is a degree-one positive homogenous and convex path-dependent measure of risk, that is $\CED_\alpha(\lambda X) = \lambda \CED_\alpha(X)$ for $\lambda>0$ and $\CED_\alpha(\lambda X +(1-\lambda)Y) \leq \lambda \CED_\alpha(X) +(1-\lambda)\CED_\alpha (Y)$ for $\lambda\in[0,1]$.
\end{prop}

\subsection{Drawdown duration}

For a fixed time horizon $T$, our main object of interest is now the temporal dimension of drawdown. One such dimension is the duration of the  drawdown process $D^{(X)}$ corresponding to the price process $X$, which measures the length of excursions of $X$ below a running maximum. Commonly referred to as the \emph{Time To Recover (TTR)} in the fund management industry, the duration captures the time it takes to reach a previous running maximum of a process for the first time. 

As before, fix a time horizon $T\in(0,\infty)$ and let $D^{(X)}=\{D^{(X)}_t\}_{t\in[0,T]}$ be the drawdown process corresponding to the stochastic process $X\in\Rspace^\infty$, and $\overline{X}$ be the running maximum of $X$. 

\begin{defn}[Peak time process]
The \emph{peak time process} $G^{(X)} = \{ G^{(X)}_t \}_{t\in[0,T]}$ is defined by
$$ G^{(X)}_t = \sup \left\{ s \in [0,t] : X_s = \overline{X}_s \right\}. $$
In words, $G^{(X)}_t$ denotes the last time $X$ was at its peak, that is the last time it coincided with its maximum $\overline{X}$ before $t$. 
\end{defn}

Note that $G^{(X)}$ is necessarily non-decreasing, consists of only linear subprocesses (more specifically, as a function of $t$, linear intervals $\{G^{(X)}_t\}_{t\in[r,s]}$ for $r<s$ are either the identity or a constant), and has jump discontinuities (under the realistic assumption that the underlying process $X$ is not monotonic). Moreover, the process $G^{(X)}$ is invariant under constant shifts or scalar multiplication of the underlying process. Similar to the drawdown process, one can show that the peak time process associated to a stochastic process $X$ does not preserve monotonicity. In other words, the peak time process corresponding to a process of larger value need not be larger. However, unlike the drawdown process which preserves convexity, peak time is not preserved under either convexity or concavity. This means that a convex combination of two processes does not result in a peak time process that is consistently smaller or greater in magnitude than the average standalone peak times of the underlying processes. 
We formally state these properties.

\begin{lem}[Properties of peak time]
Given stochastic processes $X\in\Rspace^\infty$, let $G^{(X)}$ be the corresponding peak time process for a fixed time horizon $T$. Then:
\begin{enumerate}
\item For all constant deterministic processes $C \in \Rspace^\infty$, $G_t^{(C)} = t$ for all $t\in[0,T]$.
\item For constant deterministic $C\in\Rspace^\infty$, $G_t^{(X+C)} = G_t^{(X)}$ for all $t\in[0,T]$.
\item For $\lambda >0$, $G_t^{(\lambda X)} = G_t^{(X)}$ for all $t\in[0,T]$.
\end{enumerate}
\label{lem:peaktime_prop}
\end{lem}

\begin{rem}
Note also that peak time is not necessarily preserved under monotonicity, that is $X\leq Y$ does not necessarily imply either $G_t^{(X)}\leq G_t^{(Y)}$ or $G_t^{(X)}\geq G_t^{(Y)}$ for all $t\in[0,T]$. Intuitively, the last time a process coincides with its running maximum is independent of the magnitude of the process. 
Moreover, peak time does not necessarily exhibit either quasiconvex- or quasiconcave-like behavior, that is for $\lambda\in[0,1]$, $G_t^{(\lambda X + (1-\lambda)Y)}$ is not necessarily either greater than $\min \{ G_t^{(X)}, G_t^{(Y)} \}$ or smaller than $\max \{ G_t^{(X)}, G_t^{(Y)} \}$ for all $t\in[0,T]$.
We construct a simple example showing that for $\lambda\in[0,1]$, $G_t^{(\lambda X + (1-\lambda)Y)}$ is not necessarily greater than $\min \{ G_t^{(X)}, G_t^{(Y)} \}$ for all $t\in[0,T]$. Fix a time $t\in[0,T]$ and, without loss of generality, let $G_t^{(X)} = t_0$ and $G_t^{(Y)}=t_1$ with $t_0<t_1\leq t$. Let $G_t^{(\lambda X + (1-\lambda)Y)} = t^*$ and we examine what happens if $t^*< t_0 = \min \{ G_t^{(X)}, G_t^{(Y)} \}$. In this case, we have by definition that $\lambda X + (1-\lambda)Y < \overline{\lambda X + (1-\lambda)Y}$. In particular, at $t_0$, we have
$$ \lambda X_{t_0} + (1-\lambda) Y_{t_0} < \overline{(\lambda X + (1-\lambda)Y)}_{t_0} = \lambda \overline{X}_{t_0} + (1-\lambda) \overline{Y}_{t_0} = \lambda X_{t_0} + (1-\lambda) \overline{Y}_{t_0} \ ,$$
implying that $Y_{t_0} < \overline{Y}_{t_0}$. Note that we do not have information about where $Y$ is relative to its running maximum $\overline{Y}$ before time $t_1$. This means that if $Y_{t_0} < \overline{Y}_{t_0}$, then $t^*<  \min \{ G_t^{(X)}, G_t^{(Y)} \}$, and on the other hand, if $Y_{t_0} = \overline{Y}_{t_0}$, then $t^*\geq \min \{ G_t^{(X)}, G_t^{(Y)} \}$.
One can construct a similar argument to show that $G_t^{(\lambda X + (1-\lambda)Y)}$ is not necessarily smaller than $\max \{ G_t^{(X)}, G_t^{(Y)} \}$.
\end{rem}

Probabilistically, the trajectory of the process $X$ between its peak time $G^{(X)}_t$ and its recovery time $L_t = \sup \{ s \in [t,T] : X_s = \overline{X}_s \}$ is the excursion of $X$ at its running maximum, which straddles time $t$. If $X<\overline{X}$ during this excursion, we say that $X$ is in drawdown or \emph{underwater}. 
We are now interested in $t-G^{(X)}_t$, which is the duration of this excursion.

\begin{defn}[Duration process]
Given a process $X \in \Rspace^\infty$ and time horizon $T$, the \emph{duration process}
$\delta^{(X)} = \{ \delta^{(X)}_t \}_{t\in[0,T]}$
 associated with $X$ is defined by 
$$\delta^{(X)}_t = t-G^{(X)}_t.$$
\end{defn}

The following properties are immediate consequences of Lemma \ref{lem:peaktime_prop}

\begin{lem}[Properties of duration]
Given a process $X \in \Rspace^\infty$ and time horizon $T$, the duration process $\delta^{(X)} = \{ \delta^{(X)}_t \}_{t\in[0,T]}$ satisfies the following properties on $[0,T]$:
\begin{enumerate}
\item For all constant deterministic processes $C \in \Rspace^\infty$, $\delta^{(C)} = 0$.
\item For all $X\in\Rspace^\infty$ and all constant deterministic processes $C \in \Rspace^\infty$, $\delta^{(X+C)} = \delta^{(X)}$.
\item For all $X\in\Rspace^\infty$ and $\lambda>0$, $\delta^{(\lambda X)} = \delta^{(X)}$.
\end{enumerate}
\label{lem:dur_prop}
\end{lem}

\begin{rem}
Similar to peak time, drawdown duration is not necessarily preserved under monotonicity, that is $X\leq Y$ does not necessarily imply either $\delta^{(X)}\leq \delta^{(Y)}$ or $\delta^{(X)}\geq \delta^{(Y)}$, and that drawdown duration does not necessarily exhibit either convex- or concave-like behavior, that is for $\lambda\in[0,1]$, $\delta^{(\lambda X + (1-\lambda)Y)}$ is not necessarily either greater or smaller than $\lambda\delta^{(X)}+(1-\lambda)\delta^{(Y)}$.
\end{rem}

Of particular interest now is the maximum time spent underwater within a fixed time horizon $T$, independent of the magnitude of the actual drawdown experienced by the process $X$ during this time interval. 

\begin{defn}[Maximum duration]
Given a process $X \in \Rspace^\infty$ and time horizon $T$, let $\delta^{(X)}$ be the duration process corresponding to $X$. The \emph{maximum duration} of the stochastic process $X$ is the real valued random variable defined by
$$ \delta^{(X)}_{\max} = \sup_{t\in[0,T]} \{ \delta^{(X)}_t \} . $$
\end{defn}

Maximum duration is clearly a random time in $\mathcal{T}$ defined on the same probability space as $X$ and taking values in the interval $[0,T]$. 

\begin{rem}[Duration of Maximum Drawdown]
We point out that our notion of maximum duration $\delta^{(X)}_{\max}$ differs from the duration or length of the deepest excursion below the maximum, $\mu(X)$, within the given path. Suppose that the maximum drawdown of the process $X$ occurred between times $\tau_p\in[0,T)$ (the ``peak") and $\tau_r\in[\tau_p,T]$ (the ``recovery"), where we assume for the sake of illustration that $\tau_r$ is defined, that is recovery indeed occurs within the given horizon. Note that there must be a point in time $\tau_b \in (\tau_p,\tau_r)$ where $X$ was at its minimum (the ``bottom") during the interval $(\tau_p,\tau_r)$. 
The time at which the minimum of $X$ within the interval in which the maximum drawdown occurred is given by
$$ \tau_b = \inf \{ t\in[0,T] : \mu(X) = \sup_{t\in[0,T]} D_t \}. $$
Then $\tau_p$ is the last time $X$ was at its maximum before $\tau_b$:
$$ \tau_p = \sup \{ t\in[0,\tau_b] : X_t = \overline{X}_t \}, $$
and $\tau_r$ is the first time $X$ coincides again with its rolling maximum:
$$ \tau_r = \inf \{t\in[\tau_b,T] : X_t = \overline{X}_t\} .$$
Given a process $X \in \Rspace^\infty$, the duration of the maximum drawdown of $X$ is then the random variable defined by $\delta^{(X)}_\mu = \tau_r-\tau_p$.
It is a straightforward exercise to show that $\delta^{(X)}_\mu$ satisfies the same properties that maximum duration  $\delta^{(X)}_{\max}$ satisfies.

Empirically, the two notions are closely related (see Figure \ref{maxdd_duration}), with the duration of maximum drawdown being noisier (and in fact following the actual historical maximum drawdown very closely). In studying the duration of maximum drawdown one is hence essentially analyzing the maximum drawdown itself, and this reduces to the spatial dimension of the underlying process. By considering the maximum duration, however, one is focused entirely on the temporal dimension, even though the two are correlated, as we shall see later.
\end{rem}

\begin{figure}
\centering
  \includegraphics[width=0.7\linewidth]{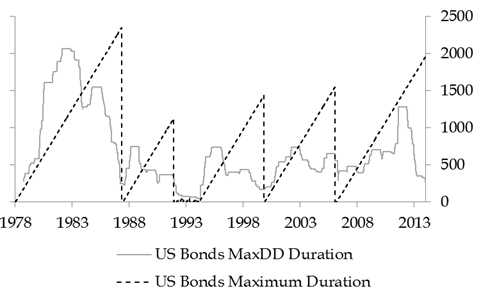}
\caption{Daily time series of historical maximum duration and the duration of maximum drawdown for US Bonds over the period 1978--2013.}
\label{maxdd_duration}
\end{figure}

\begin{rem}[Stopping time]

A second temporal dimension of drawdown that may be of interest in practice is the stopping time. In probability theory, a stopping time (also known as Markov time) is a random time whose value is interpreted as the time at which a given stochastic process exhibits a certain behavior of interest. A stopping time is generally defined by a stopping rule, a mechanism for deciding whether to continue or stop a process on the basis of the present position and past events, and which will almost always lead to a decision to stop at some finite time. It is thus completely determined by (at most) the total information known up to a certain time. 
In our context of drawdowns of investments, it may be of interest to calculate the probability that a process stays underwater for a period longer than a certain subjectively set acceptance threshold. No matter what the magnitude of the loss is, if the time to recover exceeds this threshold, one may be forced to liquidate. Given a process $X \in \Rspace^\infty$ over a time horizon $T\in(0,\infty)$ with corresponding duration process $\delta^{(X)}$, denote by $l\in(0,T]$ a subjectively set liquidation threshold. The \emph{liquidation stopping time} (LST) $\tau_L$ is defined by
$$ \tau^{(X)}_L = \inf \{ t \in [0,T] : \delta^{(X)}_t \geq l \}. $$

This stopping time $\tau_L$ hence denotes the first time the drawdown duration $\delta^{(X)}$ exceeds the pre-specified liquidation threshold $l$; that is the first time the process $X$ has stayed underwater for a consecutive period of length greater than $l$. It essentially specifies a rule that tells us when to exit a trade. Note that the decision to ``stop" at time $\tau_L$ can only depend (at most) on the information known up to that time and not on any future information. 
\footnote{From a probabilistic viewpoint, the liquidation stopping time $\tau_L$ can be identified with Parisian stopping times for a zero barrier, which was studied in \citet{Chesney1997} and \citet{Loeffen2013}. Stopping times were also introduced in actuarial risk theory in \citet{DassiosWu2009}, where the process $X$ models the surplus of an insurance company with initial capital, and the stopping time of an excursion is referred to as Parisian ruin time.}
\end{rem}


\section{Duration risk}
\label{section:duration_risk}

\begin{figure}
\centering
  \includegraphics[width=\linewidth]{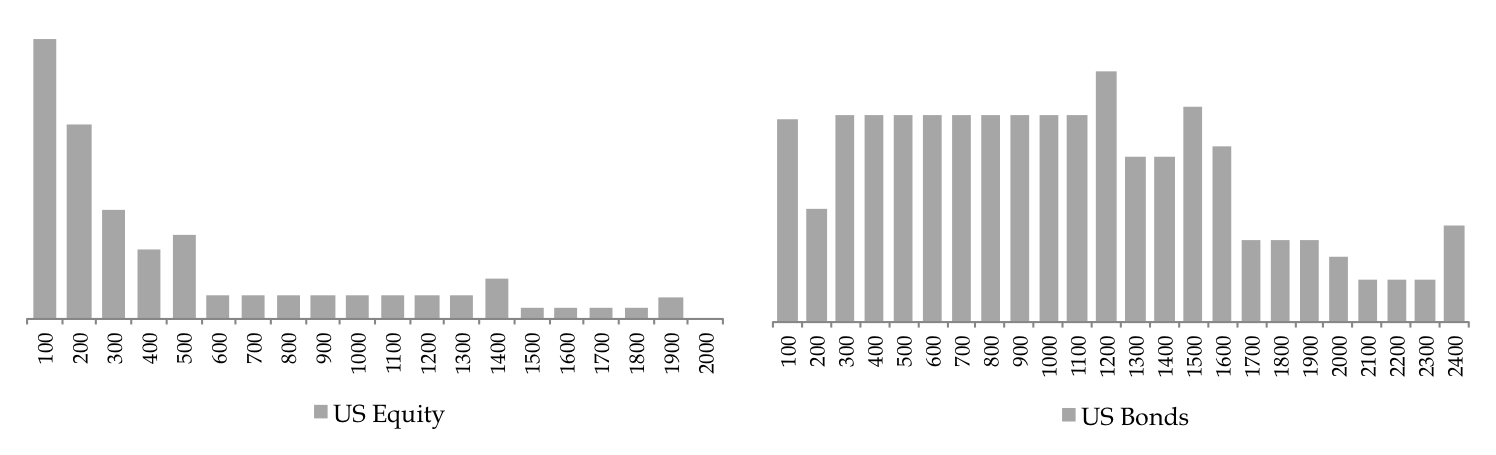}
\caption{Empirical distribution of maximum duration for daily US Equity and US Bond Indices over the period between January 1973 and December 2013, inclusive. The distribution is generated empirically as follows. From the historical time series of returns, we generate return paths of a fixed 6-month length ($n=180$ trading days) using a one-day rolling window. The maximum duration within each path is then calculated. This means that consecutive paths overlap. The advantage is that for a daily return time series of length $T$, we obtain a maximum duration series of length $T-n$, which for large $T$ and small $n$ is fairly large, too.}
\label{duration_distr}
\end{figure}

We analyze the distribution of drawdown duration, both theoretically and empirically, using temporal risk measures. Even though, in a given horizon, only a single maximum duration is realized along any given path, it is beneficial to consider the distribution from which the maximum duration is taken. By looking at this distribution, one can form reasonable expectations about the expected length of drawdowns for a given portfolio over a given investment horizon. 

Figure \ref{duration_distr} displays the empirical maximum duration distribution of US Equity and US Government Bonds over the 40-year period 1973--2013 using daily data.\footnote{The data were obtained from the Global Financial Data database. We took the
daily time series for the S\&P 500 Index and the USA 10-year Government Bond Total Return Index covering the 40-year period between January 1973 and December 2013, inclusive.} 
While maximum drawdown distributions are generally positively skewed independent of the underlying risk characteristics, which implies that very large drawdowns occur less frequently than smaller ones (see \cite{Burghardt2003} and \cite{GoldbergMahmoud2014a}), this is not necessarily the case for the distribution of maximum duration. Positive skewness is pronounced for US Equity with a value of 1.3, while it is less noticeable for US Bonds at a value of 0.4.

We can now describe the risk characteristics of maximum duration using path-dependent temporal risk measures.

\begin{defn}[Measures of duration risk]
We define the following path-dependent temporal measures of risk $\rho: \Rspace^\infty \to \R$ describing the distribution of the maximum duration $\delta^{(X)}_{\max}$ associated with a stochastic process $X\in\Rspace^\infty$:
\begin{enumerate}
\item \emph{Duration Deviation:} $\sigma_\delta : \Rspace^\infty \to \R$ is defined by
$$\sigma_\delta(X) = \sigma\left(\delta^{(X)}_{\max}\right) \ , $$
where $\sigma$ is the standard deviation.
\item \emph{Duration Quantile:} For confidence level $\alpha\in[0,1]$, the duration quantile is defined by 
$$ Q_\alpha \left(\delta^{(X)}_{\max} \right) = \inf \left\{d\in\R : \Prob(\delta^{(X)}_{\max}>d) \leq 1-\alpha \right\} . $$
\item \emph{Conditional Expected Duration}: For confidence level $\alpha\in[0,1]$, the conditional expected duration is defined by
$$ \TM_\alpha\left( \delta^{(X)}_{\max} \right) = \frac{1}{1-\alpha} \int_\alpha^1 Q_\alpha \left(\delta^{(X)}_{\max} \right)\mathrm{d}u $$
For continuous $\delta(X)$, the above amounts to
$$ \TM_\alpha\left( \delta^{(X)}_{\max} \right) = \E\left[ \delta^{(X)}_{\max} \mid \delta^{(X)}_{\max} > Q_\alpha \left(\delta^{(X)}_{\max} \right)\right]. $$
\end{enumerate}
\end{defn}

Note that each of these path-dependent risk measures is indeed temporal in the sense of Definition \ref{defn:temporal}. In each case, the path-dependent risk measure $\rho_T:\Rspace^\infty\to\R$ is the composite of a risk functional (deviation, quantile, tail mean) applied to the maximum duration with the temporal transformation mapping a stochastic process $X\in\Rspace^\infty$ to its corresponding maximum duration $\delta^{(X)}_{\max}$. We therefore immediately know that 
none of these path-dependent duration risk measures satisfies any of the coherence axioms of risk measures.

Empirically, duration risk is  consistent with the stylized fact that equities are riskier than bonds. Indeed,  Table \ref{stats} shows that Conditional Expected Duration is larger for US Equity than it is for US Bonds. On the other hand however, consistent with Figure \ref{duration_distr}, both average duration and duration deviation are considerably larger for the less risky fixed income asset than for the more risky equities asset.

\begin{table}[t]
\centering 
\begin{tabular}{l || c c c c c c } 
\hline 
 & Volatility & ES$_{0.9}$ & CED$_{0.9}$ & $\E[\delta_m^{(X)}]$ & $\sigma_{\delta}$ & CE$_{\delta,0.9}$ \\ [0.5ex] 
\hline \hline
US Equity & 18.35\%  & 2.19\% & 47\% & 456 & 489 & 1323 \\ 
US Bonds &  5.43\%  & 0.49\% & 29\% & 976 & 590 & 1070 \\
\hline
\end{tabular}
\vspace{20pt}
\caption{Single-period and path-dependent risk statistics for daily US Equity and US Bond Indices over the period between January 1973 and December 2013, inclusive. Expected Shortfall (ES), Conditional Expected Drawdown (CED) and Conditional Expected Duration (CE$_\delta$) are calculated at the 90\% confidence level. $\E[\delta_m^{(X)}]$ and $\sigma_\delta$ are the mean maximum duration and deviation of maximum duration, respectively.} 
\label{stats} 
\end{table}

\begin{rem}[Spatial versus temporal drawdown]
Even though duration is theoretically defined independent of the drawdown magnitude, there is a close relationship between the temporal and the size dimensions of a cumulative drop in portfolio value. Figure \ref{mag_time} displays the daily time series of drawdown magnitude and its duration for each of US Equity and US Government Bonds. Clearly, a drawdown's magnitude is positively correlated to its duration. Therefore, even though some smaller drawdowns can stay under water a long period of time, empirically larger drawdowns tend to come with an extended duration. In practice, minimizing the convex risk of drawdown magnitude may in fact lead to a lower overall (non-convex) duration risk.
\end{rem}

\begin{figure}
\centering
  \includegraphics[width=\linewidth]{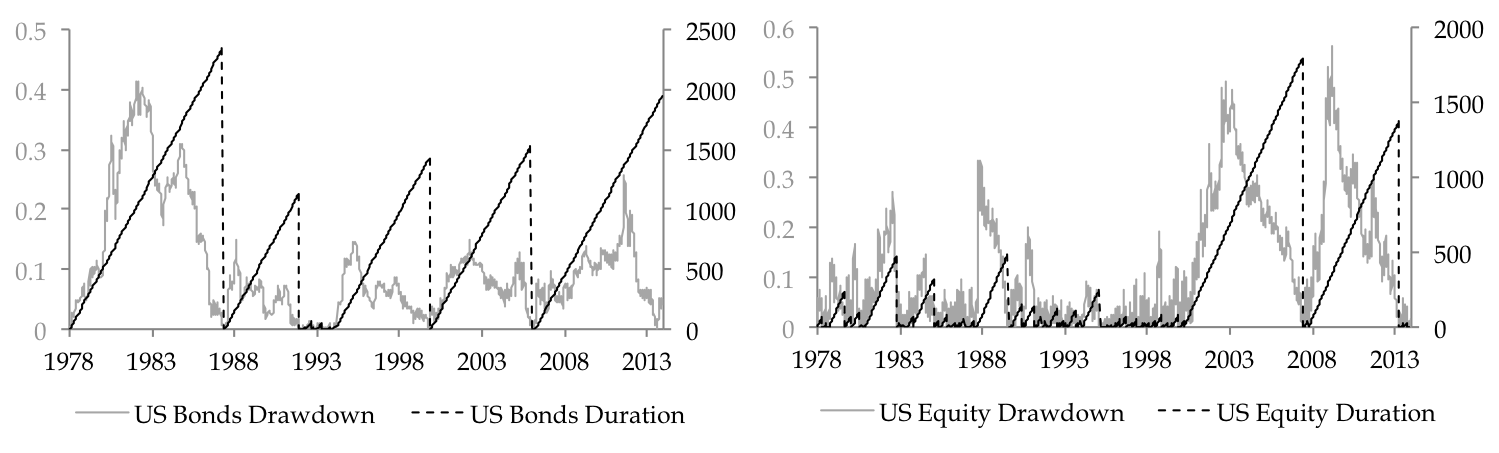}
\caption{Daily time series of historical drawdown (with scale in percentages on the left-hand side) and duration (with scale in trading days on the right-hand side) for US Bonds and US Equity over the period 1978--2013.}
\label{mag_time}
\end{figure}

\subsection{Duration risk and serial correlation}
We next show that duration risk captures temporal dependence to a greater degree than traditional one-period risk measures; this temporal dependence implies higher sensitivity to serial correlation. 

We use Monte Carlo simulation to generate an autoregressive AR(1) model:
$$ r_t = \kappa r_{t-1}  + \epsilon_t ,$$
with varying values for the autoregressive parameter $\kappa$, where $\epsilon$ is fixed to be Gaussian with variance 0.01. We then calculate volatility, Expected Shortfall, Conditional Expected Drawdown, and Conditional Expected Duration of each simulated autoregressive time series and list the results as a function of $\kappa$ (see Table \ref{kappa}). Both single-period risk measures are affected by the increase in the value of the autoregressive parameter. However, the increase is steeper for the two path dependent risk measures. We next use maximum likelihood to fit the same AR(1) model to the daily time series of US Equity and US Government Bonds on a 6-month rolling basis to obtain a time series of estimated $\kappa$ values for each asset. The correlations of the time series of $\kappa$  with the time series of 6-month rolling volatility, Expected Shortfall, Conditional Expected Drawdown, and Conditional Expected Duration are shown in Table \ref{kappa_corr}. For both assets, the correlation with the autoregressive parameter is relatively large for drawdown and duration compared to volatility and shortfall, and highest for duration.
 Finally, note that single-period and drawdown risk is consistently higher for US Equity than it is for US Bonds, whereas the opposite is true for duration risk.

These insights may impact how certain portfolio construction strategies are used in practice. Consider  for example the popular momentum trading strategy, which tends to generate relatively large returns that come equipped with some autocorrelation. Given this serial correlation, our work implies that such strategies would tend to suffer from high multi-period risk, in particular from drawdown and duration.

\begin{table}
\centering 
\begin{tabular}{l || r r r r r r r r r } 
\hline 
Kappa & 0.1 & 0.2 & 0.3 & 0.4 & 0.5 & 0.6 & 0.7 & 0.8 & 0.9 \\ [0.5ex]
\hline \hline
Volatility & 0.07  & 0.07 & 0.08 & 0.09 & 0.10 & 0.11 & 0.11 & 0.12 & 0.12 \\ 
ES$_{0.9}$ & 0.16  & 0.16 & 0.17 & 0.17 & 0.18 & 0.19 & 0.19 & 0.20 & 0.21 \\
CED$_{0.9}$ & 0.15  & 0.16 & 0.18 & 0.21 & 0.22 & 0.29 & 0.34 & 0.38 & 0.39 \\
CE$_{\delta,0.9}$ & 291  & 328 & 350 & 339 & 411 & 467 & 487 & 504 & 496 \\
\hline 
\end{tabular}
\vspace{20pt}
\caption{Volatility, 90\% Expected Shortfall, 90\% Conditional Expected Drawdown, and 90\% Conditional Expected Duration of a Monte Carlo simulated AR(1) model (with 10,000 data points) for varying values of the autoregressive parameter $\kappa$.} 
\label{kappa} 
\end{table}

\begin{table}
\centering 
\begin{tabular}{l || r r r r } 
\hline 
 & Volatility & ES$_{0.9}$ & CED$_{0.9}$ & CE$_{\delta,0.9}$\\ [0.5ex]
\hline \hline
US Equity & 0.45  & 0.52 & 0.70 & 0.81 \\ 
US Bonds &  0.32  & 0.39 & 0.67 & 0.85 \\
\hline 
\end{tabular}
\vspace{20pt}
\caption{For the daily time series of each of US Equity and US Government Bonds, correlations of estimates of the autoregressive parameter $\kappa$ in an AR(1) model with the values of the four risk measures (volatility, 90\% Expected Shortfall and 90\% Conditional Expected Drawdown, and 90\% Conditional Expected Duration) estimated over the entire period (1973--2013).} 
\label{kappa_corr} 
\end{table}


\section{Conclusion}
\label{section:conclusion}

\subsection{Summary of contributions}
Multi-period measures of risk account for the path that the value of an investment portfolio takes. In the context of probabilistic risk measures, the focus has been on the spatial dimension of path-dependent risk and not on the dimension associated with the passage of time. By incorporating the time dimension into the framework of risk measurement, one can in practice form expectations about future temporal risk. In this paper, we formalized the theory of temporal risk measures and analyzed their properties.  We have thereby introduced a new risk diagnostic complementing traditional ones and uniquely capturing the risk associated with the passage of time. We introduced the notion of temporal transformation, a spatial invariant random variable mapping of a stochastic process to a random time, and the notion of temporal acceptance family and show that these two entities correspond bijectively. A temporal risk measure is then defined as a path-dependent risk measure which can be decomposed into a temporal transformation and a risk functional. We also showed that temporal measures of risk can never be coherent in the sense of \cite{Artzner1999}.
In the second part of the paper, we studied the temporal dimension of drawdown, its duration. Drawdown duration is a widely quoted risk diagnostic in the investment management industry but has not been studied before in the context of path-dependent measures of risk. We hence formulated duration as a temporal measure of risk and derived some of its properties. We then discussed some empirical properties of duration risk. In particular, we showed that duration captures serial correlation in the returns of two major asset classes. 

\subsection{A note on estimating temporal risk in practice}
Our study of a mathematically sound notion of the temporal dimension of path-dependent risk measures essentially yields a methodology for forming expectations about future potential time-related risk.  With this mathematical setup for temporal risk measures in place, a natural ensuing question is: how do we use this formalism to make statements about future path-dependent expectations in practice? The answer lies beyond the scope of this article, and we conclude by briefly pointing towards the need for and challenges in developing a sound path-dependent risk model.

Our empirical study, even though performed for illustrative purposes, is rather simple in its estimation methodology. We have used overlapping rolling windows to generate drawdown processes, maximum drawdown and duration. Current work in progress evaluates the comparable goodness of several other estimators, for example using the block bootstrap, based on their consistency, efficiency and robustness. Moreover, note that a recovery of a given stochastic process, which is under water, may take a relatively long time compared to the available data. This is indeed the case in Figure \ref{mag_time}, where recovery of the underlying processes takes several years. This points to the challenges with respect to data availability. The sample size for generating path dependent risk measures, and in particular drawdown duration, tends to be too small. One way to overcome this issue is to reset all drawdown processes every, say, year. This means that one would forget the historical drawdown from previous years. This direction is also currently being investigated in a large scale empirical study. We would like to stress, however, that discussions with practitioners from the investment management industry --- who are the key audience on the practical side interested in using our theoretical insights --- have pointed us to the convention that clearing the drawdown memory is not favoured in practice, and hence our current estimation methodology would be chosen in practice. \footnote{I thank an anonymous referee for pointing to this suggestion.}

More generally, the key to accurate path-dependent risk forecasts is a realistic scenario generation process representing the underlying returns. Consider for example the most basic parametric Gaussian model. Despite the evidence that the Gaussian viewpoint does not yield a realistic representation of asset returns, relying on the normality assumption continues to be standard in quantitative risk measurement and reporting. There are parametric alternatives to the normal model that account for heavy tails and skewness of portfolio returns. However, the challenges of developing a flexible, robust, multi-horizon parametric model that is diverse enough to be applied to a wide range of portfolios have led to the popularity of historical simulation. We refer the reader to \cite{Cont2001} for a summary of difficulties arising from parametric modelling of equity time series.

Historical simulation, on the other hand, certainly presents challenges of its own. The methodology assumes that the past accurately represents the future, while market conditions change over time. Moreover, the data required for historical simulation may not be available. Recently developed assets may have insufficient history, and external events and economic dynamics may lead to the insignificance of an asset's history.

These issues need to be addressed in an economically sound path-dependent risk model. Compared to single-period risk measures, path-dependency introduces additional challenges. In particular, models that account for this inherent temporal dependency tend to be more complicated to estimate, simulate, and backtest in practice. On the other modeling spectrum, consider the simplest form of empirical estimation: random sampling. Such a methodology fails to account for a notion of memory in time series of returns. Memory is a stylized fact incorporating the ideas that serial autocorrelation increases during turbulent markets, volatilities change over time, and high volatility regimes have a tendency to occur immediately following large drawdowns.  Alternatives to random sampling, such as the moving block bootstrap introduced by \cite{Kuensch1989}, have certainly been developed in statistical theory. A large literature on backtesting such parametric and empirical models in the context of forecasting path-dependent risk measures does, however, not seem to exist.

\bibliographystyle{plainnat}
\bibliography{drawdown_references2}

\end{document}